\documentclass[aps,prl,twocolumn,superscriptaddress, amsmath, amssymb, amsfonts, nofootinbib]{revtex4-1}
\def\draft{1}
\usepackage{amsthm}
\usepackage{color}

\usepackage{ifthen}
\newcommand{\YSHI}[1]{\ifthenelse{\equal{\draft}{1}}{{\color{red}{#1}}}{#1}}

\newtheorem{theorem}{Theorem}
\newtheorem{proposition}[theorem]{Proposition}
\newtheorem{definition}[theorem]{Definition}
\newtheorem{lemma}[theorem]{Lemma}
\newtheorem{corollary}[theorem]{Corollary}
\newenvironment{proofof}[1]{\begin{trivlist} \item {\bf Proof #1:~~}}{\qed\end{trivlist}}

\usepackage{graphics}

\newcommand{\commentout}[1]{}

\def\Tr{\textnormal{Tr}}

\def\Supp{\textnormal{Supp }}

\begin{document}

\title{Certifying the absence of quantum nonlocality}

\author{Carl A.~Miller}
\email{carlmi@umich.edu}
\affiliation{Department of Electrical Engineering and Computer Science, University of Michigan, Ann Arbor, MI  48109, USA}

\author{Yaoyun Shi}
\email{shiyy@umich.edu}
\affiliation{Dept.~of Electrical Engineering and Computer Science, University of Michigan, Ann Arbor, MI  48109, USA}

\date{\today}

\begin{abstract}
\noindent
Quantum nonlocality is an inherently non-classical feature of quantum mechanics and manifests itself through
violation of Bell inequalities for nonlocal games. We show that in a fairly general setting, a simple extension of a nonlocal game can certify instead the {\em absence} of quantum nonlocality. Through contraposition, our result implies that  a super-classical performance for such a game ensures that a player's output is unpredictable to the other player. Previously such output unpredictability was known  with respect to a third party.
\end{abstract}


\maketitle

\paragraph*{Introduction.}  One of the most central and counterintuitive
aspects of quantum information theory is the ability for quantum players
to outperform classical players at nonlocal games.  In a nonlocal game for two players Alice
and Bob, they are given inputs
$a$ and $b$, respectively, and they produce outputs $x$ and
$y$.  The input pairs $(a, b)$ are drawn according to a fixed
distribution, and a scoring function is applied to the joint input-output tuple $(a,b,x,y)$.
When Alice and Bob use a classical strategy, they share a random variable $r$ independent of the inputs,
and decide their output deterministically from $r$ and their input.
In a quantum strategy, they share an entangled state and apply a local measurement determined
by their input. A {\em Bell inequality} upper-bounds the maximum score that a classical strategy
can achieve.
There are multi-player games for which an expected score can be achieved by quantum players
that is higher than that which can be achieved by any classical
or deterministic player. Such a violation of Bell inequality is referred to as quantum nonlocality (see \cite{Brunner:2014} for a survey of
this phenomenon).  

We ask the question: is there any way to certify the {\em absence} of quantum nonlocality?  This question needs to be more precisely formulated,
as otherwise it may appear trivially impossible. For example,
when no Bell inequality is violated, we cannot conclude that Alice and Bob did
not employ a quantum strategy. They could in principle still make use of quantum entanglement.
For example, they could measure the same observable on a maximum entangled state to produce outputs that are always anti-correlated. This input-output correlation is clearly classical yet the process is (arguably) quantum.

In this work, we call a quantum strategy {\em essentially classical}
if it is equivalent, in a sense to be made precise, to one in which all the observables of one player
commute with the shared quantum state. It appears natural to conclude that
quantum nonlocality is absent in an essentially classical strategy. Under this intepretation,
we show that the following simple extension of a nonlocal game can indeed
certify the absence of quantum nonlocality:
after the nonlocal game is played, we give Alice's input $a$ to Bob and ask him to guess what Alice's output was.
Call this second task the guessing game.
Our main theorem, stated informally, is the following. Let $A$ be Alice's local system.
\begin{theorem}[Informal] If Bob succeeds with certainty in the guessing game,  there is an isometry mapping Bob's system to $B'\otimes A'$ such that Bob's strategy for the nonlocal game involves only $B'$ and all Alice's observables commute with the reduced state on $AB'$.  Consequently, the input-output correlation is classical.
\end{theorem}

Apart from the above foundational considerations, our investigation was also motivated
by cryptography.
A useful corollary of Bell inequality violations is that quantum players
that achieve such violations are achieving \textit{certified} randomness.
Their expected score alone is enough to guarantee that their outputs
could not have been predictable to any external adversary, even when the adversary knows
the input.  This is the basis for device-independent randomness expansion \cite{col:2006, pam:2010, ColbeckK:2011,
Vazirani:dice, pironio13, fehr13, CoudronVY:2013, CY:STOC, MS-STOC14,
Universal-spot}.  When two players
play a game repeatedly and exhibit an average score above a certain threshhold,
their outputs must be highly random and can be post-processed into
uniformly random bits.  The produced uniform bits
are random even conditioned on the input bits for the game (thus ``expanding''
the input randomness).

An important and challenging question
arises: does a high
score at a nonlocal game imply that one player's output is random
\textit{to the other player}?  
Such a question is important for randomness expansion in a mutually
mistrustful scenario: suppose that Bob is Alice's adversary, and Alice
wishes to perform randomness expansion by interacting with him,
while maintaining the security of her bits against him. The contraposition of our result
implies that a violation of Bell inequality
in the nonlocal game necessarily requires that Alice's output 
expands the input randomness, with respect to Bob.

Similar problems have been studied in the literature in settings different from ours.
There has been other work showing upper bounds on the probability
that a third party can guess Alice's output after a game
(e.g., \cite{Pawlowski:2010}, \cite{Kempe:2011b}) and single-round games
have appeared
where Bob is sometimes given
\textit{only} Alice's input, and asked to produce her output
(e.g., \cite{Vazirani:fully}, \cite{Vidick:2013}).  Two recent papers
address randomness between multiple players under assumptions about imperfect storage
\cite{Kaniewski:2016, Rib:2016}.

\commentout{\YSHI{Those ``guessing'' problems are defined in  different contexts from ours, and the solutions do not
apply to our problem. In particular, earlier works on certifiable randomness
focused on randomness against a third party, instead of between the players.
Intra-party randomness is indeed proved in~\cite{Kaniewski:2016, Rib:2016}, which are independent to
ours. However, their model is significantly different:
no nonlocal game is played, and Bob's memory is either bounded, or is subject to
a noise operator before the guessing stage.}
\commentout{We believe that the novelty of our
scenario lies in the fact that we are proving that Alice's 
output is uncertain to Bob after the execution of the game (when
information has potentially been lost due to measurement).}
}

\commentout{
Our proof proceeds by relating the scenario in which
Bob tries to guess Alice's output to a scenario in which 
a third party tries to guess Alice's output.  We essentially
show that if Bob could perfectly guess Alice's output, 
then his system can be split into two subsystems and one
of these subsystem can be given to a third party (without disturbing
Bob's strategy) who can then use it to guess Alice's output.
}
\paragraph*{Preliminaries.}

For any finite-dimensional Hilbert space $V$, let
$L ( V )$ denote the vector space of linear
automorphisms of $V$.  For any $M, N \in L ( V )$,
we let $\left< M , N \right>$ denote $\Tr [ M^* N ]$.

Throughout this paper we fix four disjoint finite sets
$\mathcal{A, B, X, Y}$, which denote, respectively,
the first player's input alphabet, the second player's
input alphabet, the first player's output alphabet,
and the second player's output alphabet.
A \textit{$2$-player (input-output) correlation} is a vector
$(p_{ab}^{xy})$ of nonnegative reals,
indexed by $a, b, x, y \in \mathcal{A} \times \mathcal{B}
\times \mathcal{X} \times \mathcal{Y}$, satisfying $\sum_{xy} p_{ab}^{xy}  = 1$
for all pairs $(a, b)$, and satisfying the condition that the quantities
\begin{eqnarray}
\begin{array}{ccc}
p_a^x := \sum_y p_{ab}^{xy}, & \hskip0.2in &
p_b^y := \sum_x p_{ab}^{xy}
\end{array}
\end{eqnarray}
are independent of $b$ and $a$, respectively (no-signaling).

A $2$-player game is a pair $(q, H)$ where
\begin{eqnarray}
q \colon \mathcal{A} \times \mathcal{B} \to [0, 1 ]
\end{eqnarray}
is a probability distribution and
\begin{eqnarray}
H \colon \mathcal{A} \times \mathcal{B} \times \mathcal{X}
\times \mathcal{Y} \to [0, \infty)
\end{eqnarray}
is a function. If $q(a, b)\ne 0$ for all $a\in \mathcal{A}$ and $b\in\mathcal{B}$,
the game is said to have a {\em complete support}. The expected score associated to
such a game for a $2$-player correlation $(p_{ab}^{xy})$
is
\begin{eqnarray}
\sum_{a, b, x, y }  q ( a, b ) H ( a, b, x, y ) p_{ab}^{xy}.
\end{eqnarray}

A \textit{$2$-player strategy} is a $5$-tuple
\begin{eqnarray}\label{eqn:strategy}
\Gamma & = & ( D, E, \{ \{ R_a^x \}_x \}_a , 
\{ \{ S_b^y \}_y \}_b , \gamma )
\end{eqnarray}
such that $D, E$ are finite dimensional Hilbert spaces,
$\{ \{ R_a^x \}_x \}_a$ is a family of $\mathcal{X}$-valued
positive operator valued measures on $D$ (indexed by $\mathcal{A}$),
$\{ \{ S_b^y \}_y \}_b$ is a family of $\mathcal{Y}$-valued
positive operator valued measures on $E$,
and $\gamma$ is a density operator on $D \otimes E$.
The \textit{second player states} $\rho_{ab}^{xy}$ of $\Gamma$
is defined by
\begin{eqnarray}
\rho_{ab}^{xy} & := & \Tr_D \left[ \sqrt{R_a^x \otimes S_b^y}
\gamma \sqrt{R_a^x \otimes S_b^y} \right]
\end{eqnarray}
Define $\rho_a^x$ by the same expression with
$S_b^y$ replaced by the identity operator.  (These represent
the pre-measurement states of the second-player.)
Define $\rho:=\Tr_D(\gamma)
=\sum_x\rho_a^x$ for any $a$.

We say that the
strategy $\Gamma$ \textit{achieves} the $2$-player correlation
$(p_{ab}^{xy} )$ if $p_{ab}^{xy} =  \Tr [ \gamma ( R_a^x \otimes S_b^y ) ]$
for all $a, b, x, y$.
If a $2$-player correlation $(p_{ab}^{xy})$ can be achieved by
a $2$-player strategy then we say that it is a \textit{quantum}
correlation.

If $(p_{ab}^{xy})$ is a convex combination
of product distributions (i.e., distributions of the form
$(q_a^x ) \otimes (r_b^y )$ where $\sum_x q_a^x = 1$
and $\sum_y r_b^y = 1$) then we say that $(p_{ab}^{xy})$
is a \textit{classical} correlation.
Note that if the underlying state of a quantum strategy
is separable (i.e., it is a convex combination of bipartite
product states) then the correlation it achieves is classical.

\paragraph*{Congruent strategies.}

It is necessary to identify pairs of strategies that
are essentially the same from an operational standpoint.
We use a definition that is similar to definitions
from quantum self-testing (e.g., Definition 4 in
\cite{McKagueBQP}).

A \textit{unitary embedding}
from a $2$-player strategy
\begin{eqnarray}
\Gamma & = & ( D, E, \{ \{ R_a^x \}_x \}_a , 
\{ \{ S_b^y \}_y \}_b , \gamma )
\end{eqnarray}
to another $2$-player strategy
\begin{eqnarray}
\overline{\Gamma} & = & ( \overline{D}, \overline{E}, \{ \{ \overline{R}_a^x \}_x \}_a , 
\{ \{ \overline{S}_b^y \}_y \}_b , \overline{\gamma} )
\end{eqnarray}
is a pair of unitary embeddings $i \colon D \hookrightarrow
\overline{D}$ and $j \colon E \hookrightarrow \overline{E}$
such that $\overline{\gamma} = (i \otimes j ) \gamma (i \otimes j)^*$,
$R_a^x = i^* \overline{R}_a^x i$, and $S_b^y
 = j^* \overline{S}_b^y j$.
 
 Additionally, if $\Gamma$ is such that $D = D_1 \otimes D_2$,
and $R_a^x = G_a^x \otimes I$ for all $a, x$, then we will call
the strategy given by
\begin{eqnarray}
( D_1, E, \{ \{ G_a^x \}_a \}_x , 
\{ \{ S_b^y \}_y \}_b , \Tr_{D_2} \gamma )
\end{eqnarray} 
a \textit{partial trace} of $\Gamma$.  We can similarly
define a partial trace on the second subspace $E$
if it is a tensor product space.

We will say that two strategies $\Gamma$ and $\Gamma'$
are \textit{congruent} if there exists a sequence
of strategies $\Gamma = \Gamma_1, \ldots, \Gamma_n =
\Gamma'$ such that for each $i \in \{ 1, \ldots, n-1 \}$,
either $\Gamma_{i+1}$ is a partial trace of
$\Gamma_i$, or vice versa, or there is a unitary embedding
of $\Gamma_i$ into $\Gamma_{i+1}$, or vice versa.
This is an equivalence relation.  Note that if two strategies
are congruent then they achieve the same
correlation.

\paragraph*{Essentially classical strategies.}
We are ready to define the key concept in this paper and to
state formally our main theorem.

\begin{definition} A quantum strategy~(\ref{eqn:strategy}) is said to be {\em essentially classical}
if  it is congruent to one where $\gamma$ commutes with $R_a^x$ for all $x$ and $a$.
\end{definition}

We are interested in strategies after the application of which Bob can predict Alice's output given her input.
This is formalized as follows. If $\chi_1, \ldots, \chi_n$ are positive semidefinite
operators on some finite dimensional Hilbert space
$V$, then we say that $\{ \chi_1, \ldots, \chi_n \}$
is \textit{perfectly distinguishable} if 
$\chi_i$ and $\chi_j$ have orthogonal support for
any $i \neq j$.  This is equivalent to the condition
that there exists a projective measurement on $V$ which
perfectly identifies the state from the
set $\{ \chi_1 , \ldots, \chi_n \}$.
\begin{definition} A quantum strategy~(\ref{eqn:strategy}) {\em allows perfect guessing} (by Bob) 
if for any $a, b, y$, $\{\rho_{ab}^{xy}\}_x$ is perfectly distinguishable.
\end{definition}

\begin{theorem}[Main Theorem] \label{thm:main} If a strategy for a complete-support game allows
perfect guessing, then it is essentially classical.
\end{theorem}

(We note that the converse of the statement is not true.
This is because even in a classical strategy, Alice's output may depend on some local randomness,
which Bob cannot perfectly predict.)

Before giving the proof of this result, we note the following proposition, which taken together with 
Theorem~\ref{thm:main} implies that any strategy that permits perfect guessing yields a classical correlation.

\begin{proposition}
\label{commutingprop}
The correlation achieved by an essentially classical strategy must be classical.
\end{proposition}

\begin{proof}
We need only to consider the case that $\gamma$ commutes with
$R_a^x$ for all $a, x$.
For each $a \in \mathcal{A}$, let $V_a = \mathbb{C}^\mathcal{X}$,
and let $\Phi_a \colon L ( D ) \to L ( V_a \otimes D )$
be the nondestructive measurement defined by 
\begin{eqnarray}
\Phi_a ( T ) & = & \sum_{x \in \mathcal{X}} \left| x \right> \left< x \right|
\otimes \sqrt{R_a^x} T \sqrt{R_a^x}.
\end{eqnarray}
Note that by the commutativity
assumption, such operation leaves the state of $DE$ unchanged.

Without loss of generality, assume $\mathcal{A} = \{ 1, 2, \ldots, n \}$.
Let $\Lambda  \in L ( V_1 \otimes \ldots \otimes V_n \otimes D \otimes E )$
be the state that arises from applying the superoperators
$\Phi_1, \ldots, \Phi_n$, in order, to $\gamma$.  For any
$a \in \{ 1, \ldots, n \}$, the reduced state $\Lambda_{V_a E}$
is precisely the same as the result of taking the state
$\gamma$, applying the measurement $\{ R_a^x \}_x$
to $D$, and recording the result in $V_a$.
Alice and Bob can therefore generate the correlation $(p_{ab}^{xy})$
from the marginal state $\Lambda_{V_1 \cdots V_n E}$ alone (if
Alice possesses $V_1, \ldots, V_n$ and Bob possesses $E$).  Since this state is classical
on Alice's side, and therefore separable, the result follows.
\end{proof}

\begin{corollary}\label{maincor}If a strategy for a complete-support game allows perfect guessing,
the correlation achieved must be classical. $\qed$
\end{corollary}

\paragraph*{Proving Theorem~\ref{thm:main}.}
The proof will proceed as follows. 
First, we show that Alice's measurements $R_a := \{  R_a \}_x$ induce projective measurements
$Q_a:=\{Q_a^x\}_x$ on Bob's system.
Next, we argue that $Q_a$ commutes with Bob's own measurement $S_b:=\{S_b^y\}_y$ for any $b$.
This allows us to isometrically decompose Bob's system into two subsystems $E_1\otimes E_2$,
such that $S_b$ acts trivially on $E_2$, while $E_2$ alone can be used to predict $x$ given $a$.
The latter property allows us to arrive at the conclusion that $R_a$ commutes with $\gamma_{DE_1}$.

We will need the following lemma, which
is commonly used in studying two-player quantum strategies.
The proof was sketched in~\cite{TsirelsonNote} (see also Theorem 1 in \cite{sw:2008}).

\begin{lemma}
\label{commutelemma}
Let $V$ be a finite-dimensional Hilbert space and let
$\{ M_j \}$ and $\{ N_k \}$ be sets of positive semidefinite
operators on $V$ such that $M_j N_k = N_k M_j$ for all
$j, k$.  Then, there exists a unitary embedding
$i \colon V \hookrightarrow V_1 \otimes V_2$ and
positive semidefinite operators $\{ \overline{M}_j \}$ on
$V_1$ and $\{ \overline{N}_k \}$ on $V_2$ such that
$M_j = i^* (\overline{M}_j \otimes \mathbb{I} ) i$ and
$N_k = i^* ( \mathbb{I} \otimes \overline{N}_k ) i$ for all $j, k$.
\end{lemma}

\begin{proofof}{of Theorem~\ref{thm:main}}
Express $\Gamma$ as in (\ref{eqn:strategy}).  
Without loss of generality, we may assume
that $\Supp \rho = E$.
By the assumption that $\Gamma$ allows perfect guessing,
for any $a$, the second-player states $\{ \rho_a^x \}_x$
must be perfectly distinguishable (since otherwise
the post-measurement states $\{ \rho_{ab}^{xy} \}_x$
would not be). 
Therefore, we can find projective measurements $\{ \{ Q_a^x \}_x \}_a$ on $E$
such that
\begin{eqnarray}
Q_a^x \rho Q_a^x = \rho_a^x.
\end{eqnarray}
Note that for any fixed $a$, if the measurements $\{R_a^x \}_a$ and $\{Q_a^x \}_a$
are applied to $\gamma$, the outcome is always the same.

We have that the states
\begin{eqnarray}
\label{rho1st} \rho_{ab}^{xy} & = & \sqrt{ S_b^y } Q_a^x \rho Q_a^x \sqrt{S_b^y } \\
\label{rho2nd}
\rho_{ab}^{x'y} & = & \sqrt{ S_b^{y} } Q_a^{x'} \rho Q_a^{x'} \sqrt{S_b^y }
\end{eqnarray}
have orthogonal support for any $x \neq x'$.  Since $\Supp \rho = E$,  we have $c \mathbb{I} \leq \rho$ for some $c > 0$.
Therefore,
\begin{eqnarray}
\left< \sqrt{ S_b^y } c Q_a^x \sqrt{S_b^y } , \sqrt{ S_b^y } c Q_a^{x'} \sqrt{S_b^y } \right> & = & 0,
\end{eqnarray}
which implies, using the cyclicity of the trace function,
\begin{eqnarray}
\left\| Q_a^x S_b^y Q_a^{x'} \right\|_2 & = & 0.
\end{eqnarray}
Therefore, the measurements $\{ Q_a^x \}_x$ and $\{ S_b^y \}_y$ commute
for any $a, b$. 

By Lemma~\ref{commutelemma}, we can find a unitary embedding
$i \colon E \hookrightarrow E_1 \otimes E_2$ and such that
$S_b^y = i^* ( \overline{S}_b^y \otimes \mathbb{I} ) i$ and
$Q_a^x = i^* ( \mathbb{I} \otimes \overline{Q}_a^x ) i$, for measurements
$\{\overline{S}_b^y\}_y$ and $\{\overline{Q}_a^x\}_x$.
With 
\begin{align}
\overline{\gamma}= (\mathbb{I}_D \otimes i) \gamma (\mathbb{I}_D \otimes i^*),
\end{align} 
the strategy
$\Gamma$ embeds into the strategy
\begin{eqnarray*}
 \Gamma':=\left( D, E_1 \otimes E_2, \{ \{ R_a^x \}_x \}_a , \{ \overline{S}_b^y \otimes \mathbb{I}_{E_2} \}_y \}_b,
\overline{\gamma}\right).
\end{eqnarray*}

For any fixed $a$, the state $\overline{\gamma}$ is such that applying
the measurement $\{ R_a^x \}_x$ to the system $D$ and the measurement
$\{ \overline{Q}_a^x \}_x$ to the system $E_2$ always yields the same outcome.
In particular, if we let
\begin{eqnarray}
\tau_a^x & = & \Tr_{E_2} \left( \overline{Q}_a^x \overline{\gamma}  \right),
\end{eqnarray}
then $\Tr [ R_a^{x'} \tau_a^x ]$ will always be equal to $1$ if $x = x'$ and equal to $0$ otherwise.
Therefore $\{ R_a^{x} \}_x$ commutes with the operators $\{ \tau_a^x \}_x$, and thus also with 
their sum $\sum_x \tau_a^x = \Tr_{E_2} \overline{\gamma}$. 

Thus
if we trace out the strategy $\Gamma'$ over the system $E_2$,
we obtain a strategy (congruent to the original strategy $\Gamma$)
in which Alice's measurement operators commute with the shared state.
\end{proofof}

\paragraph*{Blind randomness expansion.}
When two players achieve a superclassical score at a nonlocal game, their
outputs must be at least partially unpredictable to an outside party, even
if that party knows the inputs that were given.  This fact
is one the bases for randomness expansion from untrusted devices
\cite{col:2006}, where a user referees a nonlocal game repeatedly with $2$
or more untrusted players (or, equivalently, $2$ or more untrusted quantum devices)
to expand a small uniformly random seed $S$
into a large output string $T$ that is uniform conditioned on $S$.  The players can exhibit arbitrary
quantum behavior, but it is assumed that they are prevented from communicating
with the adversary.  At the center of some of the
discussions of randomness expansion (e.g., \cite{pam:2010}) is the
fact that the min-entropy of the outputs of the players can be lower bounded
by an increasing function of the score achieved at the game.

The main result of this paper suggests a new protocol, \textit{blind} randomness expansion,
with even fewer trust assumptions.  Suppose that we wish to perform randomness expansion
with two untrusted players, where only the first player, not the second, can be blocked from communication
with the adversary.  In this case it is necessary to measure the unpredictability of
the first player's output with respect to the second player.  Our main result shows that,
for complete support games, any superclassical score guarantees
that the first player's output is unpredictable to the second. This matches
the ordinary randomness expansion scenario.

A natural next step is to put a lower bound on the min-entropy of the first player's output
to the second player, and here some divergences begin to appear between ordinary randomness
expansion and blind randomness expansion.
Consider the CHSH game.
If the correlation of two devices $D = (D_1, D_2)$ is
$(p_{ab}^{xy} )$, then the expected score for the CHSH game is  $\frac{1}{4} \sum_{x \oplus y = a \wedge b} p_{ab}^{xy}$.
The best possible winning probability
that can be achieved by a classical correlation
is $3/4$, while the best possible winning probability
that can be achieved by a quantum correlation is
$\frac{1}{2} + \frac{\sqrt{2}}{4} \approx 0.853 \ldots$.

Self-testing for the CHSH game \cite{mys:2012, ruv:2013,
MillerS:self-testing:2013} implies that any quantum strategy
that achieves the optimal score $\frac{1}{2} + \frac{\sqrt{2}}{4}$
is congruent to the following strategy (in which we use
the notation $\left| \theta \right> \in \mathbb{C}^2$
to denote the vector $\cos \theta \left| 0 \right> +
\sin \theta \left| 1 \right>$, and let $\Phi^+ = \frac{1}{\sqrt{2}} \left( 
\left| 00 \right> + \left| 11 \right> \right)$):
\begin{eqnarray*}
\begin{array}{rclcrcl} 
\gamma & =  & \Phi^+ \left( \Phi^+ \right)^* \\ \\
R_0^0 & = & \left| 0 \right>
\left< 0 \right| & \hskip0.5in & R_1^0 & = & \left| \frac{\pi}{4} \right>
\left< \frac{\pi}{4} \right|  \\ \\
S_0^0 & = & \left| \frac{\pi}{8} \right> \left< \frac{\pi}{8} \right| &&
S_1^0 & = & \left| - \frac{\pi}{8} \right> \left<  - \frac{\pi}{8} \right|.
\end{array}
\end{eqnarray*}
The min-entropy of Alice's output in this case --- even from the
perspective of an adversary who possess quantum side
information and knows Alice's input --- is $-\log_2 (1/2) = 1$.
(Since the state of the strategy is pure, quantum side information does not help.)

On the other hand, the second player has more information
than an external adversary.
For example, when $a = b = 0$, the second player states
are
\begin{eqnarray*}
\begin{array}{ccc}
\rho_{00}^{00}  =  \left( \frac{1}{2} + \frac{\sqrt{2}}{4} \right)
\left| \frac{\pi}{8} \right> \left< \frac{\pi}{8} \right|  & \hskip0.15in &
\rho_{00}^{01}  = \left( \frac{1}{2} - \frac{\sqrt{2}}{4} \right)
\left| \frac{5\pi}{8} \right> \left< \frac{5\pi}{8} \right|  \\ \\
\rho_{00}^{10}  =  \left( \frac{1}{2} - \frac{\sqrt{2}}{4} \right)
\left| \frac{\pi}{8} \right> \left< \frac{\pi}{8} \right|  & &
\rho_{00}^{11}  =  \left( \frac{1}{2} + \frac{\sqrt{2}}{4} \right)
\left| \frac{5 \pi}{8} \right> \left< \frac{5\pi}{8} \right|.
\end{array}
\end{eqnarray*}
If the second player wishes to guess the first player's output (given her input),
his best strategy to guess $x = 0$ if his state is $\left| \pi/8 \right>$
and to guess $x = 1$ if his state is $\left| 5 \pi / 8 \right>$.  (This is equivalent
to predicting that his own output $y$ agrees with $x$.)  Similar
results hold for other input combinations, and thus the
min-entropy of the first player's output from the second
player's perspective is $-\log_2 ( \frac{1}{2} + \frac{\sqrt{2}}{4} )  < 1$.
Thus, while one-shot blind randomness expansion is achieved
for the same scores (at complete-support games) as ordinary randomness expansion,
the certified min-entropy may be different.

\paragraph*{Further directions.}
\commentout{
We have shown that any strategy in which Alice's output is predictable
to Bob is equivalent to a classical strategy.  The flipside of this
result is that any superclassical score at a complete-support
nonlocal game generates data for Alice that is not predictable to Bob.
This suggests an approach to randomness expansion between
two mutually mistrustful parties, or blind randomness expansion.
}

A natural next step 
would be to prove a strong robust version of
Corollary~\ref{maincor} --- for example, one could attempt to prove, for the CHSH game, an nondecreasing
function $f \colon (3/4, \frac{1}{2} + \frac{\sqrt{2}}{4} ) \to \mathbb{R}_{>0}$ which 
lower bounds the min-entropy of the first player's output as a function of the score achieved
(similar to Figure 2 of \cite{pam:2010}).
Since the proof of Theorem~\ref{thm:main} relies centrally on the commutativity of certain measurements,
the notion of approximate commutativity \cite{Ozawa:2013, CoudronV:2015} may be useful
for a robust proof.

\commentout{
Our result can be seen as related to known results on
uncertainty principles in the presence of quantum side information
(see section 2.4 in~\cite{WehnerW:2010}).  One of the distinguishing features
of our result is that the quantum side information is possessed
by a party who is active rather than passive.
}

A potentially interesting aspect of Corollary~\ref{maincor} is that it contains
a notion of \textit{certified erasure} of information.  Note that in the CHSH
example above, if Bob were asked before his turn to guess Alice's output
given her input, he could do this perfectly.  (Indeed, this would be
the case in any strategy that uses a maximally entangled state
and projective measurements.)
Contrary to this, when Bob is compelled to carry out his part of the
strategy before Alice's input is revealed, he loses the ability
to perfectly guess Alice's output.  Requiring a superclassical score from Alice and Bob
amounts to forcing Bob to erase information.  Different
variants of certified erasure are a topic of current study
\cite{Unruh:2015, Kaniewski:2016, Rib:2016}.
An interesting
research avenue is determine the minimal assumptions under
which certified erasure is possible.

We also note that the scenario in which the second player
tries to guess the first player's output after computing his own
output fits the general framework of \textit{sequential
nonlocal correlations} \cite{Gallego:2014}.
In \cite{Curchod:2015} such correlations
are used for ordinary (non-blind) randomness expansion.
Another interesting avenue is to explore
how our techniques could be applied to more general
sequential nonlocal games.

\paragraph{Acknowledgments.} 
The first author thanks Jedrzej Kaniewski, Marcin Pawlowski and Stefano Pironio  for
helpful information. This research was supported in part by US NSF
Awards 1500095, 1216729, 1526928, and 1318070.

\bibliographystyle{apsrev}
\bibliography{../quantumsec}

\begin{thebibliography}{28}
\expandafter\ifx\csname natexlab\endcsname\relax\def\natexlab#1{#1}\fi
\expandafter\ifx\csname bibnamefont\endcsname\relax
  \def\bibnamefont#1{#1}\fi
\expandafter\ifx\csname bibfnamefont\endcsname\relax
  \def\bibfnamefont#1{#1}\fi
\expandafter\ifx\csname citenamefont\endcsname\relax
  \def\citenamefont#1{#1}\fi
\expandafter\ifx\csname url\endcsname\relax
  \def\url#1{\texttt{#1}}\fi
\expandafter\ifx\csname urlprefix\endcsname\relax\def\urlprefix{URL }\fi
\providecommand{\bibinfo}[2]{#2}
\providecommand{\eprint}[2][]{\url{#2}}

\bibitem[{\citenamefont{Brunner et~al.}(2014)\citenamefont{Brunner, Cavalcanti,
  Pironio, Scarani, and Wehner}}]{Brunner:2014}
\bibinfo{author}{\bibfnamefont{N.}~\bibnamefont{Brunner}},
  \bibinfo{author}{\bibfnamefont{D.}~\bibnamefont{Cavalcanti}},
  \bibinfo{author}{\bibfnamefont{S.}~\bibnamefont{Pironio}},
  \bibinfo{author}{\bibfnamefont{V.}~\bibnamefont{Scarani}}, \bibnamefont{and}
  \bibinfo{author}{\bibfnamefont{S.}~\bibnamefont{Wehner}},
  \bibinfo{journal}{Rev.~Mod.~Phys.~86, 419} \textbf{\bibinfo{volume}{86}}
  (\bibinfo{year}{2014}).

\bibitem[{\citenamefont{Colbeck}(2006)}]{col:2006}
\bibinfo{author}{\bibfnamefont{R.}~\bibnamefont{Colbeck}}, Ph.D. thesis,
  \bibinfo{school}{University of Cambridge} (\bibinfo{year}{2006}),
  \bibinfo{note}{arXiv:0911.3814}.

\bibitem[{\citenamefont{Pironio et~al.}(2010)\citenamefont{Pironio, Ac{\'\i}n,
  Massar, Boyer de~la Giroday, Matsukevich, Maunz, Olmschenk, Hayes, Luo,
  Manning et~al.}}]{pam:2010}
\bibinfo{author}{\bibfnamefont{S.}~\bibnamefont{Pironio}},
  \bibinfo{author}{\bibfnamefont{A.}~\bibnamefont{Ac{\'\i}n}},
  \bibinfo{author}{\bibfnamefont{S.}~\bibnamefont{Massar}},
  \bibinfo{author}{\bibfnamefont{A.}~\bibnamefont{Boyer de~la Giroday}},
  \bibinfo{author}{\bibfnamefont{D.~N.} \bibnamefont{Matsukevich}},
  \bibinfo{author}{\bibfnamefont{P.}~\bibnamefont{Maunz}},
  \bibinfo{author}{\bibfnamefont{S.}~\bibnamefont{Olmschenk}},
  \bibinfo{author}{\bibfnamefont{D.}~\bibnamefont{Hayes}},
  \bibinfo{author}{\bibfnamefont{L.}~\bibnamefont{Luo}},
  \bibinfo{author}{\bibfnamefont{T.~A.} \bibnamefont{Manning}},
  \bibnamefont{et~al.}, \bibinfo{journal}{Nature}
  \textbf{\bibinfo{volume}{464}}, \bibinfo{pages}{1021} (\bibinfo{year}{2010}).

\bibitem[{\citenamefont{Colbeck and Kent}(2011)}]{ColbeckK:2011}
\bibinfo{author}{\bibfnamefont{R.}~\bibnamefont{Colbeck}} \bibnamefont{and}
  \bibinfo{author}{\bibfnamefont{A.}~\bibnamefont{Kent}},
  \bibinfo{journal}{Journal of Physics A: Mathematical and Theoretical}
  \textbf{\bibinfo{volume}{44}}, \bibinfo{pages}{095305}
  (\bibinfo{year}{2011}),
  \urlprefix\url{http://stacks.iop.org/1751-8121/44/i=9/a=095305}.

\bibitem[{\citenamefont{Vazirani and Vidick}(2012)}]{Vazirani:dice}
\bibinfo{author}{\bibfnamefont{U.~V.} \bibnamefont{Vazirani}} \bibnamefont{and}
  \bibinfo{author}{\bibfnamefont{T.}~\bibnamefont{Vidick}}, in
  \emph{\bibinfo{booktitle}{Proceedings of the 44th Symposium on Theory of
  Computing Conference, {STOC} 2012, New York, {NY}, {USA}, May 19 - 22,
  2012}}, edited by \bibinfo{editor}{\bibfnamefont{H.~J.}
  \bibnamefont{Karloff}} \bibnamefont{and}
  \bibinfo{editor}{\bibfnamefont{T.}~\bibnamefont{Pitassi}}
  (\bibinfo{publisher}{ACM}, \bibinfo{year}{2012}), pp.
  \bibinfo{pages}{61--76}, ISBN \bibinfo{isbn}{978-1-4503-1245-5},
  \urlprefix\url{http://dl.acm.org/citation.cfm?id=2213977}.

\bibitem[{\citenamefont{Pironio and Massar}(2013)}]{pironio13}
\bibinfo{author}{\bibfnamefont{S.}~\bibnamefont{Pironio}} \bibnamefont{and}
  \bibinfo{author}{\bibfnamefont{S.}~\bibnamefont{Massar}},
  \bibinfo{journal}{Phys. Rev. A} \textbf{\bibinfo{volume}{87}},
  \bibinfo{pages}{012336} (\bibinfo{year}{2013}),
  \urlprefix\url{http://link.aps.org/doi/10.1103/PhysRevA.87.012336}.

\bibitem[{\citenamefont{Fehr et~al.}(2013)\citenamefont{Fehr, Gelles, and
  Schaffner}}]{fehr13}
\bibinfo{author}{\bibfnamefont{S.}~\bibnamefont{Fehr}},
  \bibinfo{author}{\bibfnamefont{R.}~\bibnamefont{Gelles}}, \bibnamefont{and}
  \bibinfo{author}{\bibfnamefont{C.}~\bibnamefont{Schaffner}},
  \bibinfo{journal}{Phys. Rev. A} \textbf{\bibinfo{volume}{87}},
  \bibinfo{pages}{012335} (\bibinfo{year}{2013}),
  \urlprefix\url{http://link.aps.org/doi/10.1103/PhysRevA.87.012335}.

\bibitem[{\citenamefont{Coudron et~al.}(2013)\citenamefont{Coudron, Vidick, and
  Yuen}}]{CoudronVY:2013}
\bibinfo{author}{\bibfnamefont{M.}~\bibnamefont{Coudron}},
  \bibinfo{author}{\bibfnamefont{T.}~\bibnamefont{Vidick}}, \bibnamefont{and}
  \bibinfo{author}{\bibfnamefont{H.}~\bibnamefont{Yuen}}, in
  \emph{\bibinfo{booktitle}{Proceedings of {APPROX} 2013 and {RANDOM} 2013}}
  (\bibinfo{publisher}{Springer}, \bibinfo{year}{2013}), vol.
  \bibinfo{volume}{8096} of \emph{\bibinfo{series}{Lecture Notes in Computer
  Science}}, pp. \bibinfo{pages}{468--483}.

\bibitem[{\citenamefont{Coudron and Yuen}(2014)}]{CY:STOC}
\bibinfo{author}{\bibfnamefont{M.}~\bibnamefont{Coudron}} \bibnamefont{and}
  \bibinfo{author}{\bibfnamefont{H.}~\bibnamefont{Yuen}}, in
  \emph{\bibinfo{booktitle}{Proceedings of the 46th Annual ACM Symposium on
  Theory of Computing}} (\bibinfo{year}{2014}), pp. \bibinfo{pages}{427--436}.

\bibitem[{\citenamefont{Miller and Shi}(2014)}]{MS-STOC14}
\bibinfo{author}{\bibfnamefont{C.~A.} \bibnamefont{Miller}} \bibnamefont{and}
  \bibinfo{author}{\bibfnamefont{Y.}~\bibnamefont{Shi}}, in
  \emph{\bibinfo{booktitle}{Proceedings of the 46th Annual ACM Symposium on
  Theory of Computing}} (\bibinfo{year}{2014}), pp. \bibinfo{pages}{417--426}.

\bibitem[{\citenamefont{Miller and Shi}(2015)}]{Universal-spot}
\bibinfo{author}{\bibfnamefont{C.~A.} \bibnamefont{Miller}} \bibnamefont{and}
  \bibinfo{author}{\bibfnamefont{Y.}~\bibnamefont{Shi}},
  \emph{\bibinfo{title}{Universal security for randomness expansion from the
  spot-checking protocol}} (\bibinfo{year}{2015}),
  \bibinfo{note}{arXiv:1411.6608}.

\bibitem[{\citenamefont{Pawlowski}(2010)}]{Pawlowski:2010}
\bibinfo{author}{\bibfnamefont{M.}~\bibnamefont{Pawlowski}},
  \bibinfo{journal}{Physical Review A} \textbf{\bibinfo{volume}{82}},
  \bibinfo{pages}{032313} (\bibinfo{year}{2010}).

\bibitem[{\citenamefont{Kempe et~al.}(2011)\citenamefont{Kempe, Kobayashi,
  Matsumoto, Toner, and Vidick}}]{Kempe:2011b}
\bibinfo{author}{\bibfnamefont{J.}~\bibnamefont{Kempe}},
  \bibinfo{author}{\bibfnamefont{H.}~\bibnamefont{Kobayashi}},
  \bibinfo{author}{\bibfnamefont{K.}~\bibnamefont{Matsumoto}},
  \bibinfo{author}{\bibfnamefont{B.}~\bibnamefont{Toner}}, \bibnamefont{and}
  \bibinfo{author}{\bibfnamefont{T.}~\bibnamefont{Vidick}},
  \bibinfo{journal}{SIAM Journal on Computing} \textbf{\bibinfo{volume}{40}},
  \bibinfo{pages}{848} (\bibinfo{year}{2011}),
  \eprint{http://dx.doi.org/10.1137/090751293},
  \urlprefix\url{http://dx.doi.org/10.1137/090751293}.

\bibitem[{\citenamefont{Vazirani and Vidick}(2014)}]{Vazirani:fully}
\bibinfo{author}{\bibfnamefont{U.}~\bibnamefont{Vazirani}} \bibnamefont{and}
  \bibinfo{author}{\bibfnamefont{T.}~\bibnamefont{Vidick}}, in
  \emph{\bibinfo{booktitle}{Proceedings of The 5th Innovations in Theoretical
  Computer Science (ITCS)}} (\bibinfo{year}{2014}),
  \bibinfo{note}{arXiv:1210.1810v2}.

\bibitem[{\citenamefont{Vidick}(2013)}]{Vidick:2013}
\bibinfo{author}{\bibfnamefont{T.}~\bibnamefont{Vidick}}, in
  \emph{\bibinfo{booktitle}{Proceedings - Annuel IEEE Symposium on Foundations
  of Computer Science (FOCS)}} (\bibinfo{year}{2013}), pp.
  \bibinfo{pages}{766--755}.

\bibitem[{\citenamefont{Kaniewski and Wehner}(2016)}]{Kaniewski:2016}
\bibinfo{author}{\bibfnamefont{J.}~\bibnamefont{Kaniewski}} \bibnamefont{and}
  \bibinfo{author}{\bibfnamefont{S.}~\bibnamefont{Wehner}},
  \emph{\bibinfo{title}{Device-independent two-party cryptography secure
  against sequential attacks}}, \bibinfo{howpublished}{arXiv:1601.06752}
  (\bibinfo{year}{2016}).

\bibitem[{\citenamefont{Riberio et~al.}(2016)\citenamefont{Riberio, Thinh,
  Kaniewski, Helsen, and Wehner}}]{Rib:2016}
\bibinfo{author}{\bibfnamefont{J.}~\bibnamefont{Riberio}},
  \bibinfo{author}{\bibfnamefont{L.~P.} \bibnamefont{Thinh}},
  \bibinfo{author}{\bibfnamefont{J.}~\bibnamefont{Kaniewski}},
  \bibinfo{author}{\bibfnamefont{J.}~\bibnamefont{Helsen}}, \bibnamefont{and}
  \bibinfo{author}{\bibfnamefont{S.}~\bibnamefont{Wehner}},
  \emph{\bibinfo{title}{Device-independence for two-party cryptography and
  position verification}}, \bibinfo{howpublished}{arXiv:1606.08750}
  (\bibinfo{year}{2016}).

\bibitem[{\citenamefont{McKague}(2015)}]{McKagueBQP}
\bibinfo{author}{\bibfnamefont{M.}~\bibnamefont{McKague}},
  \emph{\bibinfo{title}{Interactive proofs for bqp via self-tested graph
  states}}, \bibinfo{howpublished}{arXiv:1309.5675v2} (\bibinfo{year}{2015}).

\bibitem[{\citenamefont{Tsirelson}()}]{TsirelsonNote}
\bibinfo{author}{\bibfnamefont{B.}~\bibnamefont{Tsirelson}},
  \bibinfo{note}{{B}ell inequalities and operator algebras,
  \url{http://www.tau.ac.il/~tsirel/download/bellopalg.pdf}}.

\bibitem[{\citenamefont{Scholz and Werner}(2008)}]{sw:2008}
\bibinfo{author}{\bibfnamefont{V.~B.} \bibnamefont{Scholz}} \bibnamefont{and}
  \bibinfo{author}{\bibfnamefont{R.~F.} \bibnamefont{Werner}}
  (\bibinfo{year}{2008}), \bibinfo{note}{arXiv:0812.4305v1}.

\bibitem[{\citenamefont{McKague et~al.}(2012)\citenamefont{McKague, Yang, and
  Scarani}}]{mys:2012}
\bibinfo{author}{\bibfnamefont{M.}~\bibnamefont{McKague}},
  \bibinfo{author}{\bibfnamefont{T.~H.} \bibnamefont{Yang}}, \bibnamefont{and}
  \bibinfo{author}{\bibfnamefont{V.}~\bibnamefont{Scarani}},
  \bibinfo{journal}{Journal of Physics A: Mathematical and Theoretical}
  \textbf{\bibinfo{volume}{45}}, \bibinfo{pages}{455304}
  (\bibinfo{year}{2012}),
  \urlprefix\url{http://stacks.iop.org/1751-8121/45/i=45/a=455304}.

\bibitem[{\citenamefont{Reichardt et~al.}(2013)\citenamefont{Reichardt, Unger,
  and Vazirani}}]{ruv:2013}
\bibinfo{author}{\bibfnamefont{B.~W.} \bibnamefont{Reichardt}},
  \bibinfo{author}{\bibfnamefont{F.}~\bibnamefont{Unger}}, \bibnamefont{and}
  \bibinfo{author}{\bibfnamefont{U.}~\bibnamefont{Vazirani}},
  \bibinfo{journal}{Nature} \textbf{\bibinfo{volume}{496}},
  \bibinfo{pages}{456} (\bibinfo{year}{2013}).

\bibitem[{\citenamefont{Miller and Shi}(2013)}]{MillerS:self-testing:2013}
\bibinfo{author}{\bibfnamefont{C.~A.} \bibnamefont{Miller}} \bibnamefont{and}
  \bibinfo{author}{\bibfnamefont{Y.}~\bibnamefont{Shi}}, in
  \emph{\bibinfo{booktitle}{8th Conference on the Theory of Quantum
  Computation, Communication and Cryptography, {TQC} 2013, May 21-23, 2013,
  Guelph, Canada}}, edited by
  \bibinfo{editor}{\bibfnamefont{S.}~\bibnamefont{Severini}} \bibnamefont{and}
  \bibinfo{editor}{\bibfnamefont{F.~G. S.~L.} \bibnamefont{Brand{\~a}o}}
  (\bibinfo{publisher}{Schloss Dagstuhl - Leibniz-Zentrum fuer Informatik},
  \bibinfo{year}{2013}), vol.~\bibinfo{volume}{22} of
  \emph{\bibinfo{series}{LIPIcs}}, pp. \bibinfo{pages}{254--262}, ISBN
  \bibinfo{isbn}{978-3-939897-55-2}, \bibinfo{note}{full version:
  {arXiv}:1207.1819},
  \urlprefix\url{http://drops.dagstuhl.de/opus/portals/lipics/index.php?semnr=13014}.

\bibitem[{\citenamefont{Ozawa}(2013)}]{Ozawa:2013}
\bibinfo{author}{\bibfnamefont{N.}~\bibnamefont{Ozawa}},
  \bibinfo{journal}{Journal of Mathematical Physics}
  \textbf{\bibinfo{volume}{54}} (\bibinfo{year}{2013}).

\bibitem[{\citenamefont{Coudron and Vidick}(2015)}]{CoudronV:2015}
\bibinfo{author}{\bibfnamefont{M.}~\bibnamefont{Coudron}} \bibnamefont{and}
  \bibinfo{author}{\bibfnamefont{T.}~\bibnamefont{Vidick}},
  \emph{\bibinfo{title}{Proceedings of the 42nd International Colloquium on
  Automata, Languages, and Programming (ICALP)}} (\bibinfo{year}{2015}), chap.
  \bibinfo{chapter}{Interactive Proofs with Approximately Commuting Provers},
  pp. \bibinfo{pages}{355--366}.

\bibitem[{\citenamefont{Unruh}(2015)}]{Unruh:2015}
\bibinfo{author}{\bibfnamefont{D.}~\bibnamefont{Unruh}}, \bibinfo{journal}{J.
  ACM} \textbf{\bibinfo{volume}{62}}, \bibinfo{pages}{49:1}
  (\bibinfo{year}{2015}), ISSN \bibinfo{issn}{0004-5411},
  \urlprefix\url{http://doi.acm.org/10.1145/2817206}.

\bibitem[{\citenamefont{Gallego et~al.}(2014)\citenamefont{Gallego, Wurflinger,
  Chaves, Acin, and Navascues}}]{Gallego:2014}
\bibinfo{author}{\bibfnamefont{R.}~\bibnamefont{Gallego}},
  \bibinfo{author}{\bibfnamefont{L.~E.} \bibnamefont{Wurflinger}},
  \bibinfo{author}{\bibfnamefont{R.}~\bibnamefont{Chaves}},
  \bibinfo{author}{\bibfnamefont{A.}~\bibnamefont{Acin}}, \bibnamefont{and}
  \bibinfo{author}{\bibfnamefont{M.}~\bibnamefont{Navascues}},
  \bibinfo{journal}{New Journal of Physics} \textbf{\bibinfo{volume}{16}}
  (\bibinfo{year}{2014}).

\bibitem[{\citenamefont{Curchod et~al.}(2015)\citenamefont{Curchod, Johansson,
  Hoban, Wittek, and Acin}}]{Curchod:2015}
\bibinfo{author}{\bibfnamefont{F.~J.} \bibnamefont{Curchod}},
  \bibinfo{author}{\bibfnamefont{M.}~\bibnamefont{Johansson}},
  \bibinfo{author}{\bibfnamefont{M.~J.} \bibnamefont{Hoban}},
  \bibinfo{author}{\bibfnamefont{P.}~\bibnamefont{Wittek}}, \bibnamefont{and}
  \bibinfo{author}{\bibfnamefont{A.}~\bibnamefont{Acin}},
  \emph{\bibinfo{title}{Unbounded randomness certification using sequences of
  measurements}}, \bibinfo{howpublished}{arXiv:1510.03394v1}
  (\bibinfo{year}{2015}).

\end{thebibliography}
\end{document}